\documentclass[12pt,oneside]{amsart}
\usepackage{amssymb}
\usepackage{amsfonts}
\usepackage{amscd}
\usepackage[matrix, arrow, curve]{xy}
\usepackage{graphicx}
\usepackage{color}

\numberwithin{equation}{section}
\newtheorem{theorem}{Theorem}[section]

\newtheorem{proposition}[theorem]{Proposition}

\theoremstyle{definition}

\theoremstyle{remark}

\begin{document}
%%%%%%%%%%%%%%%%%%%%%%%%%%%%%%%%%%%%%%%%%%%%%%%%%
%%%%%%%%%%%%  macrodefinitions
%%%%%%%%%%%%%%%%%%%%%%%%%%%%%%%%%%%%%%%%%%%%%%%%%
%  Macros (general)
%%%%%%%%%%%%%%%%%%%%%%%
%\newcommand{\MgNekp}{\mathcal{M}_{g,N+1}^{(k,p)}} %% moduli space
%\newcommand{\M}{\mathcal{M}_{g,N+1}^{(1)}}
\newcommand{\M}{\mathcal{M}}
\newcommand{\F}{\mathcal{F}}

\newcommand{\Teich}{\mathcal{T}_{g,N+1}^{(1)}}
\newcommand{\T}{\mathrm{T}}
%%%%   temporary
\newcommand{\corr}{\bf}
\newcommand{\vac}{|0\rangle}
\newcommand{\Ga}{\Gamma}
\newcommand{\new}{\bf}
\newcommand{\define}{\def}
\newcommand{\redefine}{\def}
\newcommand{\Cal}[1]{\mathcal{#1}}
\renewcommand{\frak}[1]{\mathfrak{{#1}}}
\newcommand{\Hom}{\rm{Hom}\,}
%%%%%%%%%%%%%%%%%%%%%%%%%%%%%%%%%%%%
%   Referencing Scheme of Martin
%%%%%%%%%%%%%%%%%%%%%%%%%%
\newcommand{\refE}[1]{(\ref{E:#1})}
\newcommand{\refCh}[1]{Chapter~\ref{Ch:#1}}
\newcommand{\refS}[1]{Section~\ref{S:#1}}
\newcommand{\refSS}[1]{Section~\ref{SS:#1}}
\newcommand{\refT}[1]{Theorem~\ref{T:#1}}
\newcommand{\refO}[1]{Observation~\ref{O:#1}}
\newcommand{\refP}[1]{Proposition~\ref{P:#1}}
\newcommand{\refD}[1]{Definition~\ref{D:#1}}
\newcommand{\refC}[1]{Corollary~\ref{C:#1}}
\newcommand{\refL}[1]{Lemma~\ref{L:#1}}
\newcommand{\refEx}[1]{Example~\ref{Ex:#1}}
%%%%%%%%%%%%%%%%%%%%%%%%%%%%%%%%%%
\newcommand{\R}{\ensuremath{\mathbb{R}}}
\newcommand{\C}{\ensuremath{\mathbb{C}}}
\newcommand{\N}{\ensuremath{\mathbb{N}}}
\newcommand{\Q}{\ensuremath{\mathbb{Q}}}
\renewcommand{\P}{\ensuremath{\mathcal{P}}}
\newcommand{\Z}{\ensuremath{\mathbb{Z}}}
%%%%%%%%%%%%%%%%%%%%%%%%%%%%%%%%%%%%%%%%%%
\newcommand{\kv}{{k^{\vee}}}
%%%%%%%%%%%%%%%%%%%%%%%%%%%%%%%%%%%%%%%%%%%%%
\renewcommand{\l}{\lambda}
%%%%%%%%%%%%%%%%%%%%%%%%%%%%%%%%%%%%%%%%%%%%%%%%%%
\newcommand{\gb}{\overline{\mathfrak{g}}}
\newcommand{\dt}{\tilde d}     % Oleg
\newcommand{\hb}{\overline{\mathfrak{h}}}
\newcommand{\g}{\mathfrak{g}}
\newcommand{\h}{\mathfrak{h}}
\newcommand{\gh}{\widehat{\mathfrak{g}}}
\newcommand{\ghN}{\widehat{\mathfrak{g}_{(N)}}}
\newcommand{\gbN}{\overline{\mathfrak{g}_{(N)}}}
\newcommand{\tr}{\mathrm{tr}}
\newcommand{\gln}{\mathfrak{gl}(n)}
\newcommand{\son}{\mathfrak{so}(n)}
\newcommand{\spnn}{\mathfrak{sp}(2n)}
\newcommand{\sln}{\mathfrak{sl}}
\newcommand{\sn}{\mathfrak{s}}
\newcommand{\so}{\mathfrak{so}}
\newcommand{\spn}{\mathfrak{sp}}
\newcommand{\tsp}{\mathfrak{tsp}(2n)}
\newcommand{\gl}{\mathfrak{gl}}
\newcommand{\slnb}{{\overline{\mathfrak{sl}}}}
\newcommand{\snb}{{\overline{\mathfrak{s}}}}
\newcommand{\sob}{{\overline{\mathfrak{so}}}}
\newcommand{\spnb}{{\overline{\mathfrak{sp}}}}
\newcommand{\glb}{{\overline{\mathfrak{gl}}}}
\newcommand{\Hwft}{\mathcal{H}_{F,\tau}}
\newcommand{\Hwftm}{\mathcal{H}_{F,\tau}^{(m)}}

%%%%%%%%%%%%%%%%%%%%%%%%%%%%%%%%%%%%%%%%%%%%%%%%%%%%
\newcommand{\car}{{\mathfrak{h}}}    % Cartan subalgebra
\newcommand{\bor}{{\mathfrak{b}}}    % Borel subalgebra
\newcommand{\nil}{{\mathfrak{n}}}    % nilpotent subalgebra
\newcommand{\vp}{{\varphi}}
\newcommand{\bh}{\widehat{\mathfrak{b}}}  % Borel subalgebra of KN algebra
\newcommand{\bb}{\overline{\mathfrak{b}}}  % Borel subalgebra of KN algebra
\newcommand{\Vh}{\widehat{\mathcal V}}
\newcommand{\KZ}{Kniz\-hnik-Zamo\-lod\-chi\-kov}
\newcommand{\TUY}{Tsuchia, Ueno  and Yamada}
\newcommand{\KN} {Kri\-che\-ver-Novi\-kov}
\newcommand{\pN}{\ensuremath{(P_1,P_2,\ldots,P_N)}}
\newcommand{\xN}{\ensuremath{(\xi_1,\xi_2,\ldots,\xi_N)}}
\newcommand{\lN}{\ensuremath{(\lambda_1,\lambda_2,\ldots,\lambda_N)}}
\newcommand{\iN}{\ensuremath{1,\ldots, N}}
\newcommand{\iNf}{\ensuremath{1,\ldots, N,\infty}}

\newcommand{\tb}{\tilde \beta}
\newcommand{\tk}{\tilde \varkappa}
\newcommand{\ka}{\kappa}
\renewcommand{\k}{\varkappa}
\newcommand{\ce}{{c}}

\newcommand{\Pif} {P_{\infty}}
\newcommand{\Pinf} {P_{\infty}}
\newcommand{\PN}{\ensuremath{\{P_1,P_2,\ldots,P_N\}}}
\newcommand{\PNi}{\ensuremath{\{P_1,P_2,\ldots,P_N,P_\infty\}}}
\newcommand{\Fln}[1][n]{F_{#1}^\lambda}
\newcommand{\tang}{\mathrm{T}}
\newcommand{\Kl}[1][\lambda]{\can^{#1}}
\newcommand{\A}{\mathcal{A}}
\newcommand{\U}{\mathcal{U}}
\newcommand{\V}{\mathcal{V}}
\newcommand{\W}{\mathcal{W}}
\renewcommand{\O}{\mathcal{O}}
\newcommand{\Ae}{\widehat{\mathcal{A}}}
\newcommand{\Ah}{\widehat{\mathcal{A}}}
\newcommand{\La}{\mathcal{L}}
\newcommand{\Le}{\widehat{\mathcal{L}}}
\newcommand{\Lh}{\widehat{\mathcal{L}}}
\newcommand{\eh}{\widehat{e}}
\newcommand{\Da}{\mathcal{D}}
\newcommand{\kndual}[2]{\langle #1,#2\rangle}
\newcommand{\cins}{\frac 1{2\pi\mathrm{i}}\int_{C_S}}
\newcommand{\cinsl}{\frac 1{24\pi\mathrm{i}}\int_{C_S}}
\newcommand{\cinc}[1]{\frac 1{2\pi\mathrm{i}}\int_{#1}}
\newcommand{\cintl}[1]{\frac 1{24\pi\mathrm{i}}\int_{#1 }}
\newcommand{\w}{\omega}
\newcommand{\ord}{\operatorname{ord}}
\newcommand{\res}{\operatorname{res}}
\newcommand{\nord}[1]{:\mkern-5mu{#1}\mkern-5mu:}
\newcommand{\codim}{\operatorname{codim}}
\newcommand{\ad}{\operatorname{ad}}
\newcommand{\Ad}{\operatorname{Ad}}
\newcommand{\supp}{\operatorname{supp}}

%%%%%%%%%%%%%%%%%%%%%%%%%%%%%%%%%%%%%%%%%%%%%%%%
\newcommand{\Fn}[1][\lambda]{\mathcal{F}^{#1}}
\newcommand{\Fl}[1][\lambda]{\mathcal{F}^{#1}}
\renewcommand{\Re}{\mathrm{Re}}

\newcommand{\ha}{H^\alpha}

\define\ldot{\hskip 1pt.\hskip 1pt}
\define\ifft{\qquad\text{if and only if}\qquad}
\define\a{\alpha}
\redefine\d{\delta}
\define\w{\omega}
\define\ep{\epsilon}
\redefine\b{\beta} \redefine\t{\tau} \redefine\i{{\,\mathrm{i}}\,}
\define\ga{\gamma}
\define\cint #1{\frac 1{2\pi\i}\int_{C_{#1}}}
\define\cintta{\frac 1{2\pi\i}\int_{C_{\tau}}}
\define\cintt{\frac 1{2\pi\i}\oint_{C}}
\define\cinttp{\frac 1{2\pi\i}\int_{C_{\tau'}}}
\define\cinto{\frac 1{2\pi\i}\int_{C_{0}}}
%\define\cinttt{\frac 1{24\pi\i}\int_{C_{\tau}}}
\define\cinttt{\frac 1{24\pi\i}\int_C}
\define\cintd{\frac 1{(2\pi \i)^2}\iint\limits_{C_{\tau}\,C_{\tau'}}}
\define\dintd{\frac 1{(2\pi \i)^2}\iint\limits_{C\,C'}}
\define\cintdr{\frac 1{(2\pi \i)^3}\int_{C_{\tau}}\int_{C_{\tau'}}
\int_{C_{\tau''}}}
\define\im{\operatorname{Im}}
\define\re{\operatorname{Re}}
%\define\res{\text{res}}
\define\res{\operatorname{res}}
\redefine\deg{\operatornamewithlimits{deg}}
\define\ord{\operatorname{ord}}
\define\rank{\operatorname{rank}}
\define\fpz{\frac {d }{dz}}
\define\dzl{\,{dz}^\l}
\define\pfz#1{\frac {d#1}{dz}}

\define\K{\Cal K}
\define\U{\Cal U}
\redefine\O{\Cal O}
\define\He{\text{\rm H}^1}
\redefine\H{{\mathrm{H}}}
\define\Ho{\text{\rm H}^0}
\define\A{\Cal A}
\define\Do{\Cal D^{1}}
\define\Dh{\widehat{\mathcal{D}}^{1}}
\redefine\L{\Cal L}
\newcommand{\ND}{\ensuremath{\mathcal{N}^D}}
\redefine\D{\Cal D^{1}}
\define\KN {Kri\-che\-ver-Novi\-kov}
\define\Pif {{P_{\infty}}}
\define\Uif {{U_{\infty}}}
\define\Uifs {{U_{\infty}^*}}
\define\KM {Kac-Moody}
\define\Fln{\Cal F^\lambda_n}
%%%%%%%%%%%%%%%%%%%%
\define\gb{\overline{\mathfrak{ g}}}
\define\G{\overline{\mathfrak{ g}}}
\define\Gb{\overline{\mathfrak{ g}}}
\redefine\g{\mathfrak{ g}}
\define\Gh{\widehat{\mathfrak{ g}}}
\define\gh{\widehat{\mathfrak{ g}}}
%%%%%%%%%%%%%%%%%%%%%%%%%%
\define\Ah{\widehat{\Cal A}}
\define\Lh{\widehat{\Cal L}}
\define\Ugh{\Cal U(\Gh)}
\define\Xh{\hat X}
\define\Tld{...}
\define\iN{i=1,\ldots,N}
\define\iNi{i=1,\ldots,N,\infty}
\define\pN{p=1,\ldots,N}
\define\pNi{p=1,\ldots,N,\infty}
\define\de{\delta}

\define\kndual#1#2{\langle #1,#2\rangle}
\define \nord #1{:\mkern-5mu{#1}\mkern-5mu:}
%\define \MgN{{\Cal M}_{g,N}} %% moduli space
%\define \MgNp{{\Cal M}_{g,N}^{(p)}} %% moduli space
\newcommand{\MgN}{\mathcal{M}_{g,N}} %% moduli space
\newcommand{\MgNeki}{\mathcal{M}_{g,N+1}^{(k,\infty)}} %% moduli space
\newcommand{\MgNeei}{\mathcal{M}_{g,N+1}^{(1,\infty)}} %% moduli space
\newcommand{\MgNekp}{\mathcal{M}_{g,N+1}^{(k,p)}} %% moduli space
\newcommand{\MgNkp}{\mathcal{M}_{g,N}^{(k,p)}} %% moduli space
\newcommand{\MgNk}{\mathcal{M}_{g,N}^{(k)}} %% moduli space
\newcommand{\MgNekpp}{\mathcal{M}_{g,N+1}^{(k,p')}} %% moduli space
\newcommand{\MgNekkpp}{\mathcal{M}_{g,N+1}^{(k',p')}} %% moduli space
\newcommand{\MgNezp}{\mathcal{M}_{g,N+1}^{(0,p)}} %% moduli space
\newcommand{\MgNeep}{\mathcal{M}_{g,N+1}^{(1,p)}} %% moduli space
\newcommand{\MgNeee}{\mathcal{M}_{g,N+1}^{(1,1)}} %% moduli space
\newcommand{\MgNeez}{\mathcal{M}_{g,N+1}^{(1,0)}} %% moduli space
\newcommand{\MgNezz}{\mathcal{M}_{g,N+1}^{(0,0)}} %% moduli space
\newcommand{\MgNi}{\mathcal{M}_{g,N}^{\infty}} %% moduli space
\newcommand{\MgNe}{\mathcal{M}_{g,N+1}} %% moduli space
\newcommand{\MgNep}{\mathcal{M}_{g,N+1}^{(1)}} %% moduli space
\newcommand{\MgNp}{\mathcal{M}_{g,N}^{(1)}} %% moduli space
\newcommand{\Mgep}{\mathcal{M}_{g,1}^{(p)}} %% moduli space
\newcommand{\MegN}{\mathcal{M}_{g,N+1}^{(1)}} %% moduli space

%\define \mpt{(M,P_1,P_2,\ldots, P_N,\Pif)} %% moduli point
%\define \mpp{(M,P_1,P_2,\ldots, P_N)} %% moduli point
%\define \MgNn{{\Cal M}_{g,N}^{(1)}} %% moduli space
%\define \MgNen{{\Cal M}_{g,N+1}^{(1)}} %% moduli space
%\define \Mgo{{\Cal M}_{g,0}} %% moduli space
%\define \mptn{(M,P_1,P_2,\ldots, P_N,\Pif,z_1,\ldots,z_N,z_\infty)}
 %% moduli point
%\define \mppn{(M,P_1,P_2,\ldots, P_N,z_1,\ldots,z_N)} %% moduli point
\define \sinf{{\widehat{\sigma}}_\infty}
\define\Wt{\widetilde{W}}
\define\St{\widetilde{S}}
\newcommand{\SigmaT}{\widetilde{\Sigma}}
\newcommand{\hT}{\widetilde{\frak h}}
\define\Wn{W^{(1)}}
\define\Wtn{\widetilde{W}^{(1)}}
\define\btn{\tilde b^{(1)}}
\define\bt{\tilde b}
\define\bn{b^{(1)}}
\define \ainf{{\frak a}_\infty} %matrices with a finite number of
                                %diagonals

%
%%%%%%%%%% Olegs definitions %%%%%%%%%%%%%%%%%%%%%%%%%%%%%%%%%%%
\define\eps{\varepsilon}    % Oleg
\newcommand{\e}{\varepsilon}
\define\doint{({\frac 1{2\pi\i}})^2\oint\limits _{C_0}
       \oint\limits _{C_0}}                            % Oleg
\define\noint{ {\frac 1{2\pi\i}} \oint}   % Oleg
\define \fh{{\frak h}}     % Oleg
\define \fg{{\frak g}}     % Oleg
\define \GKN{{\Cal G}}   % affine Krichever-Novikov algebra % Oleg
\define \gaff{{\hat\frak g}}   % affine Krichever-Novikov algebra
\define\V{\Cal V}
\define \ms{{\Cal M}_{g,N}} %% moduli space
\define \mse{{\Cal M}_{g,N+1}} %% moduli space
%%%%%%%%%%%%%%%%%%%%%%%%%%%%%%%%%%%%%%
\define \tOmega{\Tilde\Omega}
\define \tw{\Tilde\omega}
\define \hw{\hat\omega}
\define \s{\sigma}
\define \car{{\frak h}}    % Cartan subalgebra
\define \bor{{\frak b}}    % Borel subalgebra
\define \nil{{\frak n}}    % nilpotent subalgebra
\define \vp{{\varphi}}
\define\bh{\widehat{\frak b}}  % Borel subalgebra of KN algebra
\define\bb{\overline{\frak b}}  % Borel subalgebra of KN algebra
\define\KZ{Knizhnik-Zamolodchikov}
\define\ai{{\alpha(i)}}
\define\ak{{\alpha(k)}}
\define\aj{{\alpha(j)}}
\newcommand{\calF}{{\mathcal F}}
\newcommand{\ferm}{{\mathcal F}^{\infty /2}}
\newcommand{\Aut}{\operatorname{Aut}}
\newcommand{\End}{\operatorname{End}}
%%%%%%%%%%%%%%%%%%%%%%%%%%%%%%%%%%%%%%%%%%%
%%%%%%%%%%%%%%%%  цвет %%%%%%%%%%%%%%%%%%%%%%%%%%%%%%%%%
\newcommand{\novoe}{{\color[rgb]{1,0,0}\bf (Новое)}}
\newcommand{\staroe}{{\color[rgb]{0,0,1}\bf (Старое)}}
\newcommand{\red}{\color[rgb]{1,0,0}}
\newcommand{\blue}{\color[rgb]{0,0,1}}
\newcommand{\viol}{\color[rgb]{1,0,1}}%%%%%%%%%%%%%%%%%%%%%%%%%%%%%%%%%%%%%%%%%%%%%%

%%%%%%%%%%%%%%%%%%%%%%%%%%%%%%%%%
%%%%%%%%%%%%%%   for laxcent
%%%%%%%%%%%%%%%%%%%%%%%%%%%%%%%%%%
\newcommand{\laxgl}{\overline{\mathfrak{gl}}}
\newcommand{\laxsl}{\overline{\mathfrak{sl}}}
\newcommand{\laxso}{\overline{\mathfrak{so}}}
\newcommand{\laxsp}{\overline{\mathfrak{sp}}}
\newcommand{\laxs}{\overline{\mathfrak{s}}}
\newcommand{\laxg}{\overline{\frak g}}
\newcommand{\bgl}{\laxgl(n)}
%%%%%%%%%%%%%%%%%%%%%%%%
\newcommand{\tX}{\widetilde{X}}
\newcommand{\tY}{\widetilde{Y}}
\newcommand{\tZ}{\widetilde{Z}}
%%%%%%%%%%%%%%%%%%%%%%%%%%%%%%%%%%%%%%%%%%
%%%%%%%%%%%%  END of macrodefinitions
%%%%%%%%%%%%%%%%%%%%%%%%%%%%%%%%%%%%%%%%%

%%%%%%%%%%%%%%%%%%%%%%%%%%%%%%%%%
%Top-Matter
%%%%%%%%%%%%%%%%%%%%%%%%%%%%%%%
%%%%%%%%%%%%%%%%%    private header  %%%%%%%%%%%%%%%%%%%%

%\large{
\title[]{Separation of variables for Hitchin systems with the structure group  $SO(4)$, on genus two curves}
\author[O.K.Sheinman]{O.K.Sheinman}
%\date{\today}
\thanks{..............................................}
\address{Steklov Mathematical Institute of the Russian Academy of Sciences}
\dedicatory{To Sergei Petrovich Novikov on the occasion of his 85-th Birthday}
\maketitle
\begin{abstract}
Sets of points giving spectral curves can be regarded as phase coordinates of Hitchin systems. We address the problem of finding out trajectories of Hitchin systems in those coordinates. The problem is being solved for the systems with structure groups $SO(4)$ and $SL(2)$ on genus 2 curves. Our method is a transfer of straight line windings of  fibers of the Hitchin map, which are given by Prymians of the spectral curves for the systems with simple classical structure groups. The transfer is carried out by means of an analog of the Jacobi inversion map, which does not exist for Prymians in general but can be defined in the two cases in question.
\end{abstract}
\tableofcontents
%%%%%%%%%%%%%%%%%%%%%%%%%%%%%%%%%%%%%%%%
\section{Introduction}
Hitchin systems \cite{Hitchin} are almost 40 years old already but we still have only one series of more or less explicitly solved systems, namely the systems with $SL(2)$ as the structure group, on genus 2 curves. Their geometrical investigation has been pioneered by B. van Geemen and E.Previato in \cite{Previato}. Further considerations were undertaken by K. Gawedzki and P. Tran-Ngoc-Bich in \cite{Gaw}. They relied on the description due to Narasimhan and Ramanan of the moduli space of rank 2 semi-stable holomorphic bundles on genus 2 curves as projectivization of the space of holomorphic theta functions  on $\C^2$. The considerations in \cite{Gaw} basically resulted in explicit expressions for the action--angle variables of the system. The authors of \cite{Gaw} also proposed a Lax representation with rational spectral parameter for the system, realized by means of $6\times 6$ matrices, nevertheless there was no attempt to obtain any explicit solution to the system in the quoted work. The last problem has been set by I.Krichever in \cite{Kr_Lax} for Hitchin systems with the structure group $GL(n)$, for an arbitrary $n$, and arbitrary compact base curve. In particular, he has  constructed a Baker--Akhieser function which enables in principle to express the solutions in terms of theta functions. As it has been noted in \cite{Sh_Bin}, there is an obstruction for similar construction in the case of simple structure groups.

In the present paper we propose an explicit solution to the Hitchin system with the structure group $SO(4)$ on a genus 2 curve. Our approach is based on the Separation of Variables technique.

Initially, Hitchin systems have been defined as integrable systems canonically related to the moduli space of holomorphic semi-stable bundles  on a compact Riemann surface. The phase space of the system is defined as a cotangent bundle towith a natural symplectic structure. By means of the Kodaira--Spencer theory the phase space is identified with the space of Higgs bundles, i.e. the holomorphic semi-stable bundles, every one equipped with a section of the corresponding bundle if fiberwise endomorphisms. Hamiltonians of the system come from invariants of the endomorphisms.

A characteristic equation of that section regarded to as a matrix depending on moduli of vector bundles, and of a point of the base Riemann surface, gives one more Riemann surface depending on moduli, and called spectral curve of the system. It is noted by A.Gorski, N Nekrasov and V.Rubtsov in \cite{GNR} that coordinates of points the spectral curve comes through give separating variables for the system.

It was noted later, in \cite{Sh_FAN_2019}, that as soon as the class of base curves is given, there arises an option of an algebraic definition of Hitchin systems, based on classification of their spectral curves, without any addressing the moduli spaces of bundles. Indeed, knowledge of the form of the spectral curve of the system, enables us to express the Hamiltonians of the of the system (i.e. coefficients of its spectral curve), and its Poisson bracket via coordinates of the points the curve passes through (see \cite{Sh_FAN_2019,Aint_sys}, and section \ref{S:so4} of the present work). Integrability of thus defined system directly follows from the results of \cite{Aint_sys}. All together gives an elementary definition of Hitchin systems, and a proof of the linearization of their trajectories on the Jacobian of the spectral curve in the case of the structure group $GL(n)$, and on the Prymian in the case of the system with a simple structure group. In othe words, trajectories can be straightened by means the Abel transform (Abel--Prym transform, resp.).

The main problem we address here is solution of a Hitchin system in the coordinates we define them, i.e. in coordinates of the points in $\C^3$ the spectral curve comes through (section \ref{S:Inv_AP}). It is a natural idea to obtain the trajectories by means operating with the inversion of the Abel transform to the straight line windings representing trajectories on the above mentioned tori. It can be done for Jacobian, and the corresponding theta function formula has been given in \cite{Sh_Bin}. In principle, the inversion of the Abel map is given by the Riemann vanishing theorem for the theta function (see beginning of section \ref{S:Inv_AP} of the present work).  In \cite{Dubr_theta}, an idea, also coming back to Riemann, has been presented how effectivise the inversion of the Abel map. It is shown there how given a point of the Jacobian to explicitly compute symmetric functions of the points of the divisor on the Riemann surface which is the preimage of that point. It is obviously maximum of what we can hope, because the points of the preimage are defined up to a permutation.

The situation with the Abel--Prym map is completely different. This map can not be reversed with help of an analog of the Riemann theorem, because the number of zeroes of the Prym theta function (of the auxiliary function constructed with its help, to be more precise) is twice as big as dimension of the Prymian \cite{Fay}. Moreover, its zeroes are in a generic position with respect to the involution on the spectral curve. It is our key observation that in the case of the structure group $SO(4)$, and of a genus 2 base curve, the spectral curve admits one more involution coming from the hyperelliptic involution of the base curve, the two involutions commute, and the Abel--Prym transform is invariant with respect to their product. The last enables us to halfen the number of points in the preimage, and thus to reverse the Abel--Prym transform in a similar way as we did that in \cite{Sh_Bin} for Jacobians.

In \refS{Hitch} we give a brief introduction to Hitchin systems, including the necessary results due to the author on separation of variables for them. More detailed presentation of the content of that Section has been given in the recent survey \cite{Sh_Bin}, and  in the earlier works by the author \cite{BorSh,Sh_FAN_2019,Aint_sys}.

In Section \ref{S:so4} we introduce the Hitchin systems with the structure group $SO(4)$, on genus 2 curves, following \cite{Sh_FAN_2019,BorSh}, i.e. with the help of the Separation of Variables technique.

In Section \ref{S:Inv_AP} we give a quadrature formula for Hitchin systems and set the problem of its reversion. For the systems with the structure group $SO(4)$ the last descends to the problem of reversion of the Abel--Prym map in presence of the second involution possessing some rare additional properties. We give a theta function formula for symmetric functions of the points in the preimage of a point of the Prymian.

In Section \ref{S:Loc_sol} we give an algorithm for obtaining local solutions in theta functions to Hitchin systems. We restrict ourselves with only local solutions for the following reason. The above presented method of reversing the Abel (the Abel--Prym, resp.) map is local in the sense that only $x$-coordinates of the points in the preimage can be found out unambiguously (up to a permutation) while the corresponding $y$-coordinates, and $\l$-coordinates are only determined up to choice of a branch of the spectral curve. However, it is sufficient for obtaining the solution with a given initial condition, because a choice of the last  fixes a choice of the corresponding branch.

Inspite the Hitchin system with the structure group $SL(2)$ on a genus 2 curve has been given an explicit description in \cite{Gaw}, its trajectories have not been found in any coordinates except in the action--angle coordinates. In Section \ref{S:SL2} we solve the problem of obtaining trajectories of $SL(2)$ systems in the same setting as for the $SO(4)$ systems.

The author would like to emphasize the role of S.P.Novikov who insisted for years on an elaboration of exact methods for solving Hitchin systems.
%%%%%%%%%%%%%%%%%%%%%%%%%%%%%%%%%%%%%%%%%%
\section{Hitchin systems: basic definitions, and separation of variables}\label{S:Hitch}

In this section we give a brief introduction to Hitchin systems, including the necessary results by the author on separation of variables for them. For the beginning, we define Hitchin systems following his work \cite{Hitchin}.

Let $\Sigma$ be a compact genus $g$ Riemann surface endowed with a complex structure, $G$ be a complex semi-simple Lie group, $\g={\mathcal Lie}(G)$, $P_0$ be a smooth principal $G$-bundle on $\Sigma$. We will refer to $\Sigma$ as to a base curve.

A \emph{holomorphic structure} on $P_0$ is a $(0,1)$-connection, i.e. a differential operator on sections of $P_0$ locally given as ${\bar\partial}+\w$ where $\w\in\Omega^{0,1}(\Sigma,\g)$, and a gluing function $\ga$ operates on $\w$ by means a \emph{gauge transformation}: $\w\to \ga\w \ga^{-1}-({\bar\partial}\ga)\ga^{-1}$.

Let $\A$ stay for the space of semi-stable \cite{Hitchin} holomorphic structures on $P_0$, $\mathcal G$ be the group of global smooth gauge transformations. The quotient space $\mathcal N=\A/{\mathcal G}$ is referred to as the moduli space of holomorphic structures on $P_0$. Below, $\mathcal N$ plays the role of \emph{configuration space} of the Hitchin system. A point of $\mathcal N$ is nothing but a holomorphic principal $G$-bundle on $\Sigma$, which will be denoted by $P$. It is known that $\dim{\mathcal N}=\dim\g\cdot(g-1)$.

The phase space of a Hitchin system is defined as the cotangent sheaf  $T^*({\mathcal N})$. Let $T_P(\mathcal N)$ denote the cotangent space at a point $P\in{\mathcal N}$. It follows from the Kodaira--Spencer theory that $T_P(\mathcal N)\simeq H^1(\Sigma,\Ad P)$. Then, by Serr duality, we have $T_P^*(\mathcal N) \simeq H^0(\Sigma,\Ad P\otimes{\mathcal K})$ where $\mathcal K$ is the canonical class on $\Sigma$, $\Ad P$ is the holomorphic vector bundle with the fiber $\g$ associated with the bundle $P$ by means of the adjoint representation of $G$. Let $(P,\Phi)$ stay for a point of $T^*(\mathcal N)$ where $P\in\mathcal N$, $\Phi\in H^0(\Sigma,\Ad P\otimes{\mathcal K})$. Sections of the sheaf $T^*(\mathcal N)$ are referred to as \emph{Higgs fields}.

Let $\chi_\d$ be a homogeneous degree $\d$ invariant polynomial on $\g$. For every $P\in {\mathcal N}$ it gives a map $\chi_\d(P) : H^0(\Sigma,\Ad P\otimes{\mathcal K}) \to H^0(\Sigma,{\mathcal K}^\d)$. If $\Phi$ is a Higgs field, then $\chi_\d(P,\Phi)=(\chi_\d(P))(\Phi(P))$, and further on $\chi_\d(P,\Phi)\in H^0(\Sigma,{\mathcal K}^\d)$. Thus every point  $(P,\Phi)$ of the phase space gets assigned with an element of the space $H^0(\Sigma,{\mathcal K}^\d)$. Let $\{\Omega^\d_j\}$ be some base in $H^0(\Sigma,{\mathcal K}^\d)$. Then $\chi_\d(P,\Phi)=\sum H_{\d,j}(P,\Phi)\Omega^\d_j$ where $H_{\d,j}(P,\Phi)$ is a scalar function on $T^*(\mathcal N)$, for every~$j$ and~$\d$, called \emph{Hitchin Hamiltonian}.
\begin{theorem}[\cite{Hitchin}]\label{T:Hitch}
Hamiltonians $\{ H_{\d,j}\}$ are Poisson commuting on $T^*(\mathcal N)$.
\end{theorem}
We'll notice right away that for an effective solution of a Hitchin system the base $\{\Omega^\d_j\}$ should be explicitly pointed out. For the last, in turn, the class of curves $\Sigma$ should be specified, which is exactly what is done below.

We choose and fix a holomorphic 1-form $\xi$ on $\Sigma$. For a given Higgs field $\Phi$ its characteristic polynomial is well-defined (independently of gluing functions in the bundle $\Ad P$) as
\begin{equation}\label{E:spec}
 \det\left(\l-\frac{\Phi}{\xi}\right)= \l^d+\sum_{j=1}^n \l^{d-d_j} r_j
\end{equation}
where $d$ is a rank of the bundle, $n$ is a rank of the Lie algebra $\g$. We define the spectral curve of the Higgs field $\Phi$ by means of the equation
\[
   \l^d+\sum_{j=1}^n \l^{d-d_j} r_j =0
\]
For classical simple Lie algebras (and for $\gl(n)$), $r_j$  stays for basis invariant polynomials of the matrix $\Phi/\xi$, except for  $\g=\so(2n)$, for which $r_n$ is a square of the base invariant, namely of Pfaffian.

For any divisor $D$ let $\mathcal{O}(D)$ be the space of meromorphic functions $f$ on the spectral curve such that $(f)+D\ge 0$. Below, $D= (\xi)$ is a divisor of the differential $\xi$. Obviously, a degree $\d$ basis invariant of the matrix $\Phi/\xi$ belongs to the space
${\mathcal O}(\d D)$. It is not difficult to prove also that the correspondence between the gauge equivalence classes of Higgs fields on the one hand side, and their spectral curves on the other hand side, is bijective \cite{Sh_Bin}. If a class of base curves is given, these remarks enable us to obtain a description of the family of spectral curves corresponding to arbitrary Higgs fields.

In what follows, we regard to $\Sigma$ as to a hyperelliptic curve with an affine part of the form $y^2=P_{2g+1}(x)$, possessing a fixed holomorphic differential $\displaystyle \xi=\frac{dx}{y}$. Then $D=2(g-1)\cdot\infty$, and a base in the space ${\mathcal O}(\d D)$ is formed by the functions
$1,x,\ldots,x^{\d(g-1)}$ and $y,yx,\ldots,yx^{(\d-1)(g-1)-2}$ \cite{Sh_FAN_2019}. Expanding the coefficients $r_j$ in the equation \refE{spec} over those bases, we obtain
\begin{proposition}[\cite{Sh_FAN_2019}]\label{P:aspec}
Affine part of the spectral curve of a Hitchin system with a simple classical Lie algebra $\g$ (or $\g=\gl(n)$) is a full intersection of two surfaces in $\C^3$: $y^2=P_{2g+1}(x)$ and $R(\l,x,y,H)=0$, where
\begin{equation}\label{E:sp_curve_eq}
   R=\l^d+\sum_{j=1}^n \left(\sum_{k=0}^{\d_j(g-1)}H_{jk}^{(0)}x^k+ \sum_{s=0}^{(\d_j-1)(g-1)-2}H_{js}^{(1)}x^sy\right)\l^{d-d_j},
\end{equation}
$\d_j$ is the order of the $j$th basis invariant (i.e. $\d_j=d_j$, except for $\g=\so(2n)$, $j=n$, when $\d_j=d_n/2=n$), $H_{jk}^{(0)}$, $H_{js}^{(1)}\in\C$. For $\g=\so(2n)$ the bracket in the last summand has to be squared.
\end{proposition}
The coefficients $H_{jk}^{(0)}$, $H_{js}^{(1)}$ regarded to as functions of a bundle $P$ exactly coincide with the Hitchin Hamiltonians. Indeed, they are obtained as coefficients of the expansion of basis invariants of the matrix $\Phi/\xi$ over bases in the space ${\mathcal O}(\d D)$, while the above definition due to Hitchin assumes an expansion of invariants of $\Phi$ over the bases in the spaces $H^0(\Sigma,{\mathcal K}^{\d})$. These are obviously equivalent procedures because  $H^0(\Sigma,{\mathcal K}^{\d})\simeq {\mathcal O}(\d D)\xi^{\otimes \d}$ for all $\d$.

%%%%%%%%%%%%%%%%%%%%%%%%%%%%%%%%%%%%%%%
%%%%%%%%%%%%%%%%%%%%%%%%%%%%%%%%%%%%%%%%%%%%%%
Every curve can be given by a set of points it passes through. We show that giving this way the spectral curve we can completely define a Hitchin system. The idea goes back to \cite{GNR,BT2}, however it is necessary to specify the class of base curves to implement it. For the class of hyperelliptic curves it has been done in \cite{Aint_sys,Sh_FAN_2019}.

Observe that the number of coefficients in every summand in \refE{sp_curve_eq} is equal to $(2\d_j-1)(g-1)$. In the case of a simple Lie algebra $\g$, by wellknown Kostant identity, $\sum_{j=1}^n(2\d_j-1)=\dim\g$ (recall that $\d_j$ stay for the degrees of basis invariants of the Lie algebra $\g$). Denote the total number of coefficients in the equation of the spectral curve by $h$, we obtain $h=\dim\g(g-1)$. For $\g=\gl(n)$, taking account of the degree 1 invariant, we have $h=n^2(g-1)+1$ .

Giving $h$ points $\ga_i=(x_i, y_i, \l_i) ~(i = 1, \ldots, h)$, where $y_i^2 = P_{2g+1}(x_i)$ for all $i$, we can express the Hamiltonians via those points from the requirement of passing the curve through them:
\begin{equation}\label{E:sep}
  R(x_i, y_i, \l_i,H)=0,\quad i=1,\ldots,h,
\end{equation}
where every equation contains only one triple $(x_i, y_i, \l_i)$. Conventionally, such kind of equations are referred to as separating relations, while  $x_i, y_i, \l_i$ are referred to as separating variables. We emphasize that for $\g=\gl(n), \sln(n), \spn(2n), \so(2n+1)$ the system of equations \refE{sep} is linear in $H$, while in the case $\g=\so(2n)$ it is quadratic. Nevertheless, for $n=2$ it can be effectively resolved (Proposition \ref{P:PBor} below).

It is known \cite{GNR} that the action--angle variables $(I,\phi)$ of a Hitchin system can be defined from the relations
\begin{equation}\label{E:act_ang}
   \l dz=\sum_{a=1}^{h} I^a\Omega_a,\quad
   \phi_a=\sum_{i=1}^h \int^{\ga_i}\Omega_a,
\end{equation}
where $z=\int^{\ga}\widehat{\xi}$ is a (quasiglobal) coordinate on the spectral curve, $\widehat{\xi}$ is a pull-back of the differential $\xi$ on the spectral curve, $\Omega_a$ are normalized holomorphic differentials on it. In the case of a simple Lie algebra $\g$ the equation of the spectral curve contains only even degrees of $\l$, hence it is invariant with respect to the involution $\l\to -\l$. In this case $\Omega_a$ are normalized holomorphic Prym differentials. Observe that then $h=\dim\g(g-1)$ exactly is equal to to the dimension of the Prym variety of the spectral curve.

In the action-angle variables the symplectic form has the canonical form:
\begin{equation}\label{E:du}
  \w=\sum_{a=1}^h dI^a\wedge d\phi_a .
\end{equation}
From the second of the relations \refE{act_ang} we have $d\phi_a=\sum_{i=1}^h \Omega_a(\ga_i)$, hence
\begin{equation}\label{E:Sympl1}
 \w = \sum_{a=1}^h \d I^a\wedge \sum_{i=1}^h\Omega_a(\ga_i).
\end{equation}
The external differentiation of the first of the relations \refE{act_ang} gives $d\l\wedge dz=\sum_{a=1}^{h} dI^a\wedge \Omega_a$ which implies $\sum_{i=1}^h d\l_i\wedge dz_i=\sum_{i=1}^h\sum_{a=1}^{h} dI^a\wedge \Omega_a(\ga_i)$. The last expression coincides with the right hand side of \refE{Sympl1}, hence $\w=\sum_{i=1}^h d\l_i\wedge dz_i$. Since $dz=\widehat\xi$ where $\xi=\frac{dx}{y}$, in the above separation variables we have
\begin{equation}\label{E:Sympl2}
 \w= \sum_{i=1}^h d\l_i\wedge  \frac{d x_i}{y_i}.
\end{equation}
The corresponding Poisson bracket is equal to
\begin{equation}\label{E:Poisson}
       \{ \l_i,x_j\}=\d_{ij} y_i
\end{equation}
(where the brackets of all other pairs of variables vanishes).

Relations \refE{sep} and \refE{Sympl2} completely define a Hitchin system on a hyperelliptic curve, independently of its original definition. Is it possible to prove integrability of such kind a system not addressing \refT{Hitch}? We will do that based on the following proposition.
Consider a system of equations
\begin{equation}\label{E:non_sys}
F_i(H_1,\ldots,H_h,\l_i,z_i)=0,\quad i=1,\ldots,h
\end{equation}
where $F_i$ are smooth functions. Then $H_j$, $j=1,\ldots,h$ are defined as functions of $(\l_1,\ldots,\l_h,z_1,\ldots,z_h)$ by the relations \refE{non_sys}. We stress that for every $i=1,\ldots,h$ the function $F_i$ is explicitly depending on only one pair of variables $\l_i,z_i$.

Consider a Poisson bracket of the following form on $\C^{2h}$:
\begin{equation}\label{E:Puass}
\{ f,g\}=\sum\limits_{j=1}^hp_j\left(\frac{\partial f}{\partial \l_j}\frac{\partial g}{\partial z_j}-\frac{\partial g}{\partial \l_j}\frac{\partial f}{\partial z_j}\right),
\end{equation}
where $p_j=p_j(\l_j,z_j)$ are smooth functions of only one pair of variables.
\begin{proposition}[\cite{Aint_sys}]\label{P:comm}
In a domain where $\rank\frac{\partial(F_1,\ldots,F_{k-1},F_{k+1}\ldots,F_h)} {\partial(H_1,\ldots,H_{k-1},H_k,H_{k+1}\ldots,H_h)}=n-1$,
$H_1,\ldots,H_h$ commute with respect to every Poisson bracket of the above form.
\end{proposition}
Indeed, by differentiating the equations \refE{non_sys}, for $i\ne k$, first in $\l_k$, and then in $z_k$, we find that the vectors $\mathbf{x_1}=(\frac{\partial H_1}{\partial \l_k}, \ldots, \frac{\partial H_h}{\partial \l_k})$ and $\mathbf{x_2}=(\frac{\partial H_1}{\partial z_k}, \ldots, \frac{\partial H_h}{\partial z_k})$ satisfy the same system of linear equations $\frac{\partial F_i}{\partial H_1}x_1+ \ldots + \frac{\partial F_i}{\partial H_h}x_h=0$, $i\ne k$. By the condition on the rank of the system these vectors are linear dependent, hence their Pl\"{u}cker coordinates vanish: $\begin{vmatrix}
  \frac{\partial H_p}{\partial \l_k} & \frac{\partial H_q}{\partial \l_k}\\
  \frac{\partial H_p}{\partial z_k} & \frac{\partial H_q}{\partial z_k}
                           \end{vmatrix}=0
$
for all $p,q=1,\ldots,h$.

Particular cases of \refP{comm} have been formulated in a number of papers starting from the first half of the 20th century (for a somewhat incomplete bibliography see \cite{Aint_sys}), and were given complicated proofs. In the context of symplectic geometry, when all $p_i$ are not equal to $0$, \refP{comm} is equivalent to the statement that tangent vector fields to a Lagrangian submanifold commute (in our case the Lagrangian submanifold is equal to the product of curves given by the equations  \refE{non_sys} for fixed $H_j$). The last has been used to prove integrability of the systems given by equations of the form  \refE{non_sys} in \cite{Tsiganov}, for example. In our form, \refP{comm} is stronger, and is true, for example, in the case when only one of the coefficients $p_i$ is nonzero.

Coming back to Hitchin systems on hyperelliptic curves, observe that their Poisson bracket \refE{Poisson} satisfies the conditions of \refP{comm}, hence their integrability is proven.

\section{Hitchin systems with the structure group $SO(4)$, on genus 2 curves}\label{S:so4}
We define the system by means of its spectral curve. According to \refP{aspec} we give the curve by two equations in $\C^3$:
\begin{equation}\label{E:spec_so4}
    R(\l, x, y, H)=\lambda^4 + \lambda^2p + q^2 = 0,\quad \text{and}\quad y^2=P_5(x)
\end{equation}
where $p=H_0 + xH_1 + x^2H_2$, $q=H_3 + xH_4 + x^2H_5$, $P_5=x^5+\ldots$ is a degree 5 polynomial. The curve is singular. Its singular points are given by the following system of equations completed with the equations \refE{spec_so4}:
\begin{equation}\label{E:origso4}
\left \{
\begin{array}{c}
     R'_{\l}(\l, x, y, H) = \l(4\l^2 + 2p)=0,\\
     R'_{x}(\l, x, y, H) = \l^2\cdot p'_{x} + 2q\cdot q'_{x} = 0.
     \end{array}
\right.
\end{equation}
In a generic position, the curve has four singular points given by $\l=0$, $q = 0$.

For $\l^2=-p/2$, and generic $H$, the second of the equations  \refE{origso4} is not satisfied, hence we obtain  branch points. By \refE{spec_so4} they satisfy the following system of equations:
\begin{equation}\label{E:br_p}
\left \{
\begin{array}{c}
     \l^2 = -\frac{p}{2}, \\
     q = \pm \frac{p}{2}.
     \end{array}
\right.
\end{equation}
Takin account of symmetries in $y$ and $\l$, we obtain 16 branch points.

We define the phase space of the Hitchin system with the structure group $SO(4)$ on the curve $y^2=P_5(x)$ as the 12-dimensional space of sets
$\{ (\l_i,x_i,y_i)|i=1,\ldots,6 \}$ where every triple $(\l_i,x_i,y_i)$ satisfies to the second of equations \refE{spec_so4}. The Poisson bracket is defined as $\{\l_i,x_i\}=y_i$ ($i=1,\ldots,6$), while the brackets between all other pairs of variables are set to zero. The set of six Hamiltonians of the system $H=\{ H_1,\ldots, H_6\}$ is defined via phase variables from the system of equations $R(\l_i,x_i,y_i,H)=0$ ($i=1,\ldots,6$) which are nothing but separating relations. Integrability of the constructed this way system had been proven in the last section (see also \cite{Aint_sys}, \cite{Sh_FAN_2019} and \cite{BorSh}). For Hitchin systems with the structure group $SO(2n)$, the separating relations are quadratic in Hamiltonians. However, for $n=2$ they can be explicitly resolved thanks to the following statement.
\begin{proposition}[\cite{PBor}]\label{P:PBor}
For $\g = \so(4)$, $g=2$ the system of separation relations descends to an algebraic equation of degree four with one unknown, hence is solvable in radicals.
\end{proposition}
According to \cite{Hitchin}, trajectories of a Hitchin system with the structure group $SO(2n)$ get linear on the Prymian of the normalized spectral curve. We pass to the description of the last.

For a singular curve, the Riemann--Hurwitz theorem gives a genus of the curve. Define the genus $\widehat{g}$, then for our spectral curve
\[
   2\widehat{g} - 2 = 4(2g - 2) + 16
\]
which gives $\widehat{g} = 13$ for $g=2$.

The normalized curve is a 4-fold branch cover of the base curve. If it is regarded as a cover of Riemann sphere then it is a 8-fold cover whose every sheet contains 10 branch points (by means of six of them a genus 2 curve is glued from 2 spheres, while another four serve for gluing the genus 2 sheets together). The normalized curve is drown schematically at Figure \ref{so4_pic}.
%%%%%%%%%%%%%%%%%%%%%%%%%%%%%%%%%%%%%
%\input{so4}
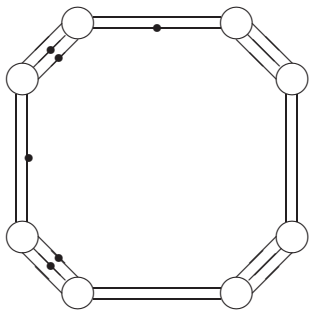
\begin{figure}[h]
\begin{picture}(200,130)
\unitlength=0.6pt
\thicklines
\thinlines
%нижняя половина
\put(55,20){\circle{20}}
\put(20,55){\circle{20}}
\put(47,28){\line(-1,1){20}}
\put(55,30){\line(-1,1){25}}
\put(45,20){\line(-1,1){25}}
\put(155,20){\circle{20}}
\put(190,55){\circle{20}}
\put(163,28){\line(1,1){20}}
\put(155,30){\line(1,1){25}}
\put(165,20){\line(1,1){25}}
\put(65,23){\line(1,0){80}}
\put(65,16){\line(1,0){80}}
%
% верхняя половина
\put(55,190){\circle{20}}
\put(20,155){\circle{20}}
\put(47,182){\line(-1,-1){20}}
\put(55,180){\line(-1,-1){25}}
\put(45,190){\line(-1,-1){25}}
\put(155,190){\circle{20}}
\put(190,155){\circle{20}}
\put(162,183){\line(1,-1){20}}
\put(155,180){\line(1,-1){25}}
\put(165,190){\line(1,-1){25}}
\put(65,187){\line(1,0){80}}
\put(65,194){\line(1,0){80}}
\put(23,65){\line(0,1){80}}
\put(16,65){\line(0,1){80}}
\put(193,65){\line(0,1){80}}
\put(186,65){\line(0,1){80}}
\put(105,187){\circle*{5}}
\put(24,105){\circle*{5}}
\put(43,168){\circle*{5}}
\put(38,173){\circle*{5}}
\put(43,42){\circle*{5}}
\put(38,37){\circle*{5}}

\end{picture}
\caption{Normalized sprctral curve}\label{so4_pic}
\end{figure}
%%%%%%%%%%%%%%%%%%%%%%%%%%%%%%%%%%%%%%%
The spectral curve \refE{spec_so4} is invariant with respect to the involution $\tau_1:\l\to-\l$. The last is represenred by by a rotation of the Figure \ref{so4_pic} around its center by an angle~$\pi$. On the normalized spectral curve, the involutions operates without fixed points. The singular points of the spectral curve are fixed but their preimages after normalization are permuted by the involution. At the Figure, the preimages are located in the middles of the tubes corresponding to the horizontal segments, two points per tube. The normalization map can be thought as gluing of points on the opposite horizontal segments.

Below, we need an explicit description of the Abel--Prym map of our spectral curve. It is given by the following relation
\[
   \A(\l_1,x_1,y_1\ldots,\l_6,x_6,y_6)=\sum_{i-1}^6\int_\infty^{(\l_i,x_i,y_i)}\w ,
\]
where $\w$ is the set of six normalized Prym differentials on the spectral curve, organized in a $6\times 1$ (a column). A base of holomorphic Prym differentials is given by the following list \cite{BorSh}:
\begin{equation}\label{E:dprym}
\begin{aligned}
       &\omega^{(0)}_i = \frac{x^{i-1}q(x)dx}{y\l(4\l^2 + 2p(x))},~ i = 1, 2, 3,\\
       &\omega^{(1)}_i = \frac{\l^2x^{i-4}dx}{y\l(4\l^2 + 2p(x))},~ i = 4, 5, 6.
\end{aligned}
\end{equation}
As it follows from the results of \cite{BorSh} (see also \cite{Dubr_RegCh}), normalization of the differentials \refE{dprym} amounts in computation of their integrals along the contours connecting pairs of branch points presented as fat points on Figure \ref{so4_pic}. As it is pointed out above, the branch points are explicitly given by equations \refE{br_p}, hence the normalized differentials can be found out explicitly also, as well as the action--angle coordinates.
%%%%%%%%%%%%%%%%%%%%%%%%%%%%%%%%%%%%%%%%%%%%%%%%%%%%

%%%%%%%%%%%%%%%%%%%%%%%%%%%%%%%%%%%%%%%%%%%%%%%%%%%%%%%%%%%%%%
\section{The quadrature formula, and its reversion}\label{S:Inv_AP}
%%%%%%%%%%%%%%%%%%%%%%%%%%%%%%%%%%%%%%%%%%%%%%%%%%%%%%%%%%%%%%
In the action--angle variables $(I,\phi)$, the symplectic form $\w$ has a canonical form \refE{du}. Hence the trajectories of motion are straight lines in these variables:
\begin{equation}\label{E:trajec}
  I^a=const,\quad \phi_a=\d_a^bt_b+\phi_{a,0}
\end{equation}
where $t_b$ is the time corresponding to the action variable $I^b$. As one can observe from the second of the relations \refE{act_ang}, the transformation $(\ga_1,\ldots,\ga_h)\to\phi$ is nothing but the Abel--Prym transformation. For this reason the trajectories are straight line windings of the Jacobian of the spectral curve for the systems with the structure group $GL(n)$, and of the prymian for the systems with simple structure groups, in particular, for $SO(4)$. Plugging the expression \refE{act_ang} for $\phi_a$  in \refE{trajec}, we will obtain the following quadrature formula for trajectories:
\begin{equation}\label{E:quad}
  \sum_{i=1}^h \int^{\ga_i}\Omega_a=\d_a^bt_b+\phi_{a,0}.
\end{equation}
We set the problem of finding out the trajectories in the coordinates $x_i,y_i,\l_i$ of the points $\ga_i$. To this goal, the Abel--Prym transform in \refE{quad} should be reversed. For the Abel transform, the solution of the problem of inversion "in principle"\ is given by the following theorem due to Riemann. Let $\theta$ be the Riemann theta function, $K$ be the vector of Riemann constants. For a generic point $\phi$ of the Jacobian, the $\A^{-1}(\phi)=D$ where $D=P_1+\ldots+P_g$, $|D|=\{ P_1,\ldots,P_g\}$ is the set of zeroes of the function $F(P)=\theta(\A(P)-K+\phi)$ on the Riemann surface dissected along a previously chosen basis cicles. It is important that $\deg D=\dim Jac$ (and equal to the genus of the Riemann surface).

Unlike the previous case, if $\theta$ is the Prym theta function, $Prym$ is the Prymian of the spectral curve, $\A$ is the Abel--Prym map, and $\phi\in Prym$, then $\deg D=2\dim Prym$ \cite[Corollary 5.6]{Fay}. In other words, the number of zeroes of the function $F$ is twice as big as $\dim Prym$, and, morover, they are in a generic position with respect to the involution. We must conclude that in general no solution to the inversion problem for Prymians can be obtained this way. However, assume that the Riemann surface admits a second involution commuting with the first one, and such that the Prym differentials are skew-symmetric to both involutions.  Denote the involutions by $\tau_1$ and $\tau_2$. Then the Riemann theorem gives $h$ pairs of points ($h=\dim Prym$), such that every pair is invariant with respect to $\tau_1\tau_2$. Note that the situation of existing the two commuting involutions exactly takes place for the Hitchin systems with the structure group $SO(4)$ on genus 2 Riemann surfaces. Indeed, we can take $\tau_1:\,\l\to -\l$,  $\tau_2:\, y\to -y$.

We address now the problem of an effective reversion of the Abel--Prym map in case of two involutions. Basically, we follow ideas by B.Dubrovin in \cite{Dubr_theta} (going back to Riemann) developed for an explicit reversion of the Abel map. We set $\tau=\tau_1\tau_2$. Due to skew-symmetry of Prym differentials with respect to both $\tau_1$ and $\tau_2$, the differentials are $\tau$-invariant. Hence $\A$ is  $\tau$-invariant too, which implies $\tau$-invariance of~$F$. For any meromorphic $\tau$-invariant function $f$ on $\widehat\Sigma$ we consider $\s_f(\phi)=\sum_{P\in |D|} f(P)$ where $D$ is the zero divisorof the function $F$, $|D|=\rm{support}(D)$. By $\tau$-invariance of $F$, the set of its zeroes $|D|$ is $\tau$-invariant too, hence generically it is representable in the form $|D|=\{ P_1,\ldots ,P_{2h}\}$ where $h=\dim Prym$, $P_i=\tau P_{i+h}$ for $i=1,\ldots,h$. Then $\s_f(\phi)=2\sum_{i=1}^h f(P_i)$. Assuming $f$ to have no pole except at infinity, we begin with the following relation by Dubrovin \cite[eq. 11.23]{Dubr_RegCh} (slightly different formulated in \cite[eq. (2.4.29)]{Dubr_theta}):
\begin{equation}\label{E:Dubr}
  \s_f(\phi)=c-\sum_{Q\in\pi^{-1}(\infty)}\res_Q f(P)d\ln\theta(\A(P)-\phi-K)
\end{equation}
where $c$ is independent of $\phi$. The proof by Dubrovin holds true for the Prym theta function, since it relies on the theorem on residues and quasiinvariance relations for theta functions only. Below, we choose $f(\l,x,y)=x^k$. Setting $P_i=(\l_i,x_i,y_i)$, and denoting $\s_f$ by $\s_k$, we observe that
\begin{equation}\label{E:sigm}
  \s_k(\phi)=2(x_1^k+\ldots +x_h^k),
\end{equation}
which means that $\s_k(\phi)$ is a doubled $k$th Newton polynomial in $x_1,\ldots ,x_h$.
Hence, the relation \refE{Dubr} can be written down as
\begin{equation}\label{E:Dubr_k}
  \s_k(\phi)=c-\sum_{Q\in\pi^{-1}(\infty)}\res_Q x^kd\ln\theta(\mathcal{A}(P)-\phi-K).
\end{equation}
Coming back to the case $SO(4)$ and genus 2, we set $h=6$. Since $(d\mathcal{A})_i=\w_i$, we will obtain that
\[
  d\ln\theta(\mathcal{A}(P)-\phi-K)=\sum_{i=1}^6 (\partial_i\ln\theta(\mathcal{A}(P)-\phi-K))\w_i,
\]
where $\w_i$ are normalized differentials \refE{dprym} , $\partial_i$ denotes the derivative in the $i$th argument ($i=1,\ldots,6$). Choosing an arbitrary ${Q_0\in\pi^{-1}(\infty)}$ as a base point of the Abel--Prym transform, in a neghborhood of $Q_0$ we may regard to $\mathcal{A}(P)$ as to a small quantity, and expand $(\ln\theta(\mathcal{A}(P)-\phi-K))_i$ in a Tailor series. After that, we only have to find out a sum of order $z^{2k-1}$ summands of the just obtained expansion where $z$ is a local parameter in the neighborhood  of $Q_0$. Clearly, having been multiplied by $x^k=z^{-2k}$ that sum will give the residue at the point $Q_0$ in \refE{Dubr_k}:
\begin{equation}\label{E:theta_sigma}
    \res_{Q_0} x^kd\ln\theta(\mathcal{A}(P)-\phi-K) = \sum_{i=1}^6\sum_{1\le |j|\le 2k-1} \varkappa_i^jD^j\partial_i\ln\theta(-\phi-K),
\end{equation}
where (for $h=6$) $j=(j_1,\ldots,j_6)$, $|j|=j_1+\ldots+j_6$,
\begin{equation}\label{E:theta_kappa}
  D^j=\frac{1}{j_1!\ldots j_6!}\frac{\partial^{|j|}}{\partial\phi_1^{j_1}\ldots\partial\phi_6^{j_6}},
  \quad \varkappa_i^j=\sum_{l_i+\sum_{s=1}^6\sum_{p=1}^{j_s} l_{sp}=2k-1} \varphi_i^{(l_i)}\prod_{s=1}^6\prod_{p=1}^{j_s}\frac{\varphi_s^{(l_{sp}-1)}}{l_{sp}},
\end{equation}
$l_s$ and $\varphi_s^{(l_s)}$ are defined from the relation $\A_s(P)=\sum_{l_s\ge 1}\frac{\varphi_s^{(l_s)}}{l_s}z^{l_s}$ ($P=P(z)$). A computation of the residue in an arbitrary point $Q\in\pi^{-1}(\infty)$ is similar, except that the expansion in the Tailor series should be performed in the quantity $\A(P)-\A(Q)$ which only results in a shift of the argument of the theta function by $\A(Q)$ in the right hand side of the relation \refE{theta_sigma}. We finally obtain
\begin{equation}\label{E:theta_sigma_1}
    \s_k(\phi)=const- \sum_{Q\in\pi^{-1}(\infty)}\sum_{i=1}^6\sum_{1\le |j|\le 2k-1} \varkappa_i^jD^j\partial_i\ln\theta(\A(Q)-\phi-K),
\end{equation}
The functions $\s_k(\phi)$, $k=1,\ldots,6$ provide a full set of symmetric functions of $x$-coordinates of the points in $\A^{-1}(\phi)$. They define $x_1,\ldots,x_6$ up to a permutation. To find out the corresponding $y_i$, $\l_i$ we plug $x_i$ into the system of equations \refE{spec_so4} which is easy resolvable with respect to $y$ and $\l$.
\begin{proposition}\label{P:loc_rev}
Let $\ga_0=\{ (x_1^0,y_1^0,\l_1^0),\ldots,(x_6^0,y_6^0,\l_6^0) \}$be a point of the phase space, $\phi_0=\A(\ga_0)$. Also, let $(x_1,\ldots,x_6)$ be a solution of the system of equations \refE{sigm} (where $k=1,\ldots,6$), $y_i$ be a solution of the second of the equations \refE{spec_so4} after the substitution $x=x_i$, $\l_i$ be a solution of the first of the equations  \refE{spec_so4} after the substitution $x=x_i$, $y=y_i$ such that $(x_i,y_i,\l_i)$ belongs to the same branch of the spectral curve as  $(x_i^0,y_i^0,\l_i^0)$. Then $\ga=\{ (x_1,y_1,\l_1),\ldots,(x_6,y_6,\l_6) \}=\A^{-1}(\phi)$ outside the branch points.
\end{proposition}
\begin{proof}
The Proposition is nothing but a statement of local invertibility of the map $\A$. Observe that $d\A=(\w_i(P^j))$ is a nondegenerate matrix in general, hence $\A$ is indeed locally invertible.
\end{proof}
%%%%%%%%%%%%%%%%%%%%%%%%%%%%%%%%%%%%%%%%%%%%%%%%%%%%

%%%%%%%%%%%%%%%%%%%%%%%%%%%%%%%%%%%%%%%%%%%%%%%%%%%%
\section{Local solutions}\label{S:Loc_sol}
Proposition \ref{P:loc_rev} enables us to explicitly find out local solutions to a Hitchin system, namely solutions in the domains of the phase space having the form of a product of  $h$ open subsets of the spectral curve, each one belonging to one of its branches.

We recall that for simple classical groups Hitchin trajectories are straight line windings of Prymians \cite{Hitchin}. In the action-angle coordinates $I,\phi$ such windings on the Prymian $I=const$ are given by equations of the form $\phi_a=\frac{\partial f}{\partial I^a}t+\phi_{0a}$, where $f$ is a function, $t$ is the time corresponding to the Hamiltonian $f(I)$ by virtue of the Noether's theorem.

We address the problem (having been set in the previous Section) of finding out the Hitchin trajectory with the Hamiltonian $f(I)$ passing through a point $\ga_0$ at $t=0$, in the coordinates $\{ (x_1,y_1,\l_1),\ldots,$ $(x_6,y_6,\l_6) \}$. We propose the following algorithm of resolving the problem:
\begin{itemize}
\item[$1^\circ$]
   Find out the action-angle coordinates of $\ga_0$ as follows:
   $\phi_0=\A(\ga_0)$; $I=AH$ where $A$ is the transition matrix from the base \refE{dprym} of unnormalized Prym differentials to the base $\Omega_a$ of normalized ones, coefficients $H$ in \refE{spec_so4} are obtained from the equations \refE{sep};
\item[$2^\circ$]
     Find out $\s_k(\phi)$ ($k=1,\ldots,6$) by means of the relations  \refE{theta_sigma}.
\item[$3^\circ$]
Then find out $x_1,\ldots,x_6$ as functions of $\phi$ from the relations \refE{sigm}, up to permutation, and plug there  $\phi_a=\frac{\partial f}{\partial I^a}t+\phi_{0a}$, ;
\item[$4^\circ$]
     For $x_i$ ($i=1,\ldots,6$) find out $y_i$, $\l_i$ by means the procedure described in  \refP{loc_rev}.
\end{itemize}
%%%%%%%%%%%%%%%%%%%%%%%%%%%%%%%%%%%%%%%%%%%%%%%%%%%%%%%%%%%%%%%%%
\section{The Hitchin system with the structure group  $SL(2)$ on a genus 2 curve}\label{S:SL2}
A Hitchin system with the algebra $\sln(2)$ on a genus $g$ hyperelliptic curve has  $N = 3(g - 1)$ independent Hamiltonians. For $g=2$ we have $N=3$. The Lie algebra $\sln(2)$ has the only invariant of the second order, hence the spectral curve has the form
\begin{equation}\label{specsl2}
R(\l, x, y) = \l^2 + r_2(x) = \l^2 + (H_0 + xH_1 + x^2H_2) = 0
\end{equation}
in this case. It is a smooth curve, because the equations on singular points ($R'_{\l} = 0$  and $R'_{x}=0$) amount in $r_2(x) = 0$ and $(r_2(x))'_x = 0$, and are incompatible unless  $r_2(x)$ has a multiple root.

The spectral curve is a two fold branch covering of the base curve, and the branch points are given by solutions of the equation $R'_{\l} = 0$ which amounts in $r_2(x) = 0$. By the symmetry in  $y$, there is  four branch points. A genus $\hat{g}$ of the spectral curve can be found out by the Riemann--Hurwitz formula:
\[
2\hat{g} - 2 = 2(2g - 2) + 4 .
\]
For $g = 2$ we obtain $\hat{g} = 5$.

In turn, the base curve is a two fold branch covering of the Riemann sphere, hence the spectral curve can be regarded to as a  4-fold branch covering of the sphere having 8 branch points on every sheet (the 6 points serve to glue the base curve of the spheres, and two more serve to glue two copies of the base curve). Thus the corresponding Riemann surface has the form given at Figure \ref{sl2_pic}, where the circles denote the Riemann spheres (that is, the sheets of the covering of the sphere), while every segment corresponds to the gluing along a contour conecting a pair of branch points (we may think of these segments as of tubes connecting the corresponding spheres):
%%%%%%%%%%%%%%%%%%%%%%%%%%%%%%%%%%%%%
% \input{sl2}
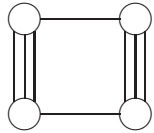
\begin{figure}[h]
\begin{picture}(150,50)
\unitlength=0.6pt
\thicklines
\thinlines
\put(10,20){\circle{20}}
\put(10,80){\circle{20}}
\put(10,30){\line(0,1){40}}
\put(3,28){\line(0,1){44}}
\put(16,28){\line(0,1){44}}
\put(80,20){\circle{20}}
\put(80,80){\circle{20}}
\put(80,30){\line(0,1){40}}
\put(73,28){\line(0,1){44}}
\put(86,28){\line(0,1){44}}
%\put(0,55){$P_1$}
%
\put(20,20){\line(1,0){50}}
\put(20,80){\line(1,0){50}}
\end{picture}
\caption{Spectral curve in the case $\sln(2)$, $g=2$}\label{sl2_pic}
\end{figure}

%%%%%%%%%%%%%%%%%%%%%%%%%%%%%%%%%%%%%
The spectral curve is invariant with respect to the involution $\l\to -\l$ (a reflection with respect to the vertical axis at the Figure). Fixed points of the involution coincide with the branch points given by the equation $\l=0$ (that is, with the branch points of the spectral curve over the base one; at the Figure, they correspond to the horizontal segments).

According to the results of \cite{BorSh}, base of the Prym differentials  is given by $\omega_i = (x^idx)/(2y\l)$,  $i = 0, 1, 2$.

For the reversion of the Abel--Prym map in a similar way to \refS{Inv_AP}, it is sufficient to compute the functions $\s_k(\phi)$ for $k=1,2,3$ by means the expression \refE{theta_sigma}. We restrict ourselves with the computation of $\s_1(\phi)$ here.
\begin{equation}\label{E:theta_sigma1}
\begin{aligned}
    \s_1(\phi)&=const- \sum_{i=1}^3\sum_{j_1+j_2+j_3= 1} \varkappa_i^jD^j\partial_i\ln\theta(-\phi-K)\\
    &= const- \sum_{i=1}^3\left( \varkappa_i^{(1,0,0)}\frac{\partial}{\partial\phi_1} +
    \varkappa_i^{(0,1,0)}\frac{\partial}{\partial\phi_2} +
    \varkappa_i^{(0,0,1)}\frac{\partial}{\partial\phi_3}
       \right)\partial_i\ln\theta(-\phi-K)
\end{aligned}
\end{equation}
To begin with, we consider the first summand in the brackets. Since $j_1=1$, $j_2=j_3=0$, it follows $l_i+l_{11}=1$ for the summation range in \refE{theta_kappa}, which follows from  $l_i+\sum_{s=1}^h\sum_{p=1}^{j_s} l_{sp}=1$,  and since $l_{11}\ge 1$ we have $l_i=0$, $l_{11}=1$. As a result, we obtain
\[
\varkappa_i^{(1,0,0)}=\sum_{l_i+l_{11}=1} \varphi_i^{(l_i)}\frac{\varphi_1^{(l_{11}-1)}}{l_{11}} =
\varphi_i^{(0)}\varphi_1^{(0)}.
\]
Similarly, if the unit is in the $s$th position in $j$ then $\varkappa_i^j=\varphi_i^{(0)}\varphi_s^{(0)}$, and (taking account of $\partial_i=-\partial/\partial\phi_i$) we obtain
\begin{equation}\label{E:theta_sig1}
\begin{aligned}
  \s_1(\phi)&=const+\sum_{i=1}^3\sum_{s=1}^3 \phi_i^{(0)}\phi_s^{(0)} \frac{\partial}{\partial\phi_i}\frac{\partial}{\partial\phi_s}\ln\theta(-\phi-K)\\
  &= const+ \left( \sum_{i=1}^3 \phi_i^{(0)}\frac{\partial}{\partial\phi_i}
            \right)^2\ln\theta(-\phi-K).
\end{aligned}
\end{equation}
This is exactly the same expression as the one in \cite{Dubr_theta}, obtained there for the Riemann theta function.
%%%%%%%%%%%%%%%%%%%%%%%%%%%%%%%%%%%%%%%%%%%%%%%%%%%%
\section{Concluding remarks}\label{S:concl}
The Hitchin systems with the structure groups $SO(4)$ and $SL(2)$ on genus 2 curves are very specific in a sense that their spectral curves possess two commuting involutions $\l\to -\l$ and $y\to -y$, with the same space of Prym differentials. It follows from the results of  \cite{Sh_FAN_2019,Aint_sys} that there is no other such Hitchin system with a simple classical structure group on a hyperelliptic Riemann surface. While it is easy to find out a spectral curve inheriting the involution  $y\to -y$ from the base curve, the spaces of Prym differentials of the two involutions generically never coincide, except for the above two cases.
%%%%%%%%%%%%%%%%%%%%%%%%%%%%%%%%%%%%%%%%%%%%%%%%%%%%
%%%%%%%%%%%%%%%%%%%%%%%%%%%%%%%%%%%%%%%%%%%%%%%%%%%%

\bibliographystyle{amsalpha}

\end{document}
%%%%%%%%%%%%%%%%%%%%%%%%%%%%%%%%%%%%%%%%%%%%%%%%
%%   THE END
%%%%%%%%%%%%%%%%%%%%%%%%%%%%%%%%%%%%%%%%%%%